%% file: main.tex
\documentclass[conference,a4paper]{IEEEtran}

\usepackage[utf8]{inputenc} 
\usepackage[T1]{fontenc}
\usepackage{url}              
\usepackage{cite}             

\usepackage[cmex10]{amsmath}
\usepackage{amssymb,amsfonts,amsthm}
\usepackage{bbm}
\usepackage{hyperref}
\usepackage{color}
\usepackage[subtle,tracking = normal, charwidths = tight,mathspacing=normal,wordspacing=normal]{savetrees}
\usepackage{mleftright}       
\mleftright                   

\usepackage{graphicx}         
\usepackage{booktabs}         

\usepackage{amsthm}
\newtheorem{theorem}{Theorem}
\newtheorem{lemma}{Lemma}
\newtheorem{corollary}{Corollary}
\newtheorem{proposition}{Proposition}
\newtheorem{definition}{Definition}

\begin{document}


\title{Learning Channel Codes from Data: Performance Guarantees in the Finite Blocklength Regime}

\author{%
  \IEEEauthorblockN{Neil Irwin Bernardo\IEEEauthorrefmark{1}\IEEEauthorrefmark{2}, Jingge Zhu\IEEEauthorrefmark{1}, and Jamie Evans\IEEEauthorrefmark{1}}
  \IEEEauthorblockA{\IEEEauthorrefmark{1}Department of Electrical and Electronic Engineering, University of Melbourne, Melbourne, Australia\\
  \IEEEauthorrefmark{2}Electrical and Electronics Engineering Institute, University of the Philippines Diliman, Quezon City, Philippines }
  \IEEEauthorblockA{
      Email: neil.bernardo@eee.upd.edu.ph, \{jingge.zhu,  jse\}@unimelb.edu.au
   }
   
}

\maketitle

\begin{abstract}
  This paper examines the maximum code rate achievable by a data-driven communication system over some unknown discrete memoryless channel in the finite blocklength regime. A class of channel codes, called \emph{learning-based channel codes}, is first introduced. Learning-based channel codes include a learning algorithm to transform the training data into a pair of encoding and decoding functions that satisfy some statistical reliability constraint. Data-dependent achievability and converse bounds in the non-asymptotic regime are established for this class of channel codes. It is shown analytically that the asymptotic expansion of the bounds for the maximum achievable code rate of the learning-based channel codes are tight for sufficiently large training data.
  
  
  
\end{abstract}

\section{Introduction}
\label{section:introduction}
\input{introduction.tex}

\section{Problem Formulation}
\label{section:problem_formulation}
\input{problem_formulation.tex}

\section{Non-asymptotic Performance Bounds}
\label{section:main_result}
\input{results.tex}

\section{Asymptotic Results}
\label{section:normal_approx}
\input{normal_approx.tex}

\section{Summary and Future Work}
\label{section:summary}
\input{summary.tex}

\section*{Acknowledgements}\label{section-acknowledgement}
\input{acknowledgement.tex}

\enlargethispage{-1.4cm} 
\bibliographystyle{IEEEtran}
\bibliography{references}
\newpage

\begin{appendices}
\section{Proof of Theorem \ref{theorem:RCU_with_learning}}\label{proof:achievability}
\input{appendix_A.tex}

\section{Proof of Theorem \ref{theorem:metaconverse_learning}}\label{proof:converse}
\input{appendix_B.tex}
\section{Proof of Theorem \ref{theorem:normal_approx_achievability}}\label{proof:normal_approx_achievability}
\input{appendix_C.tex}

\section{Proof of Corollary \ref{corollary:normal_approx_general}}\label{proof:normal_approx_general}
\input{appendix_D.tex}
\end{appendices}

\end{document}

%% file: introduction.tex
The emergence of new application scenarios in the 5$^{\mathrm{th}}$  generation (5G) of cellular technology, namely massive machine-type communication (mMTC) and ultra-reliable low-latency communication (uRLLC), has sparked interest in the fundamental tension between blocklength and reliability \cite{Durisi:2016}. The framework developed by Polyanskiy, Poor, and Verdu \cite{Polyanskiy:2010} is perhaps the most widely adopted approach to obtain tight bounds for the maximum achievable rate of a channel in the finite blocklength regime. This framework, however, requires \emph{a priori} knowledge of the communication channel.

Dropping the assumption of a perfectly known channel, some papers have studied the synergy between channel coding and channel estimation in different fading environments using the mismatched decoding framework \cite{Liva:2017,Ostman:2019}. Under the mismatched decoding framework, the instantaneous channel state information (CSI) is first estimated and the result is used to aid the decoding operation. The performance of the mismatched decoder is then evaluated with respect to the true fading distribution. Hence, the underlying channel distribution should be known to facilitate the analysis. Similarly, achievability results for the mismatched decoding capacity of discrete memoryless channels in the finite blocklength regime \cite{Scarlett:2014} require knowledge of the underlying distribution to numerically evaluate the error exponents and relevant information-theoretic quantities. Such approach may not be suitable for a \emph{model-deficit} communication problem, i.e. aside from the fact that the sender and receiver do not know the probabilistic relationship between the input and output of the channel, the information about the true channel is also not available in the analysis.


One method to circumvent the lack of channel model is to adopt a data-driven approach, i.e. machine learning and deep learning techniques are used on a finite collection of channel input-output pairs to either learn specific blocks of a communication system or learn the end-to-end communication process. There is a vast literature on the application of machine learning methods and deep neural networks (DNN) to various communication problems \cite{Simeone:2018,Cammerer:2020,Bjornson:2020,Shlezinger:2021,Shen:2022}. Deep learning-based channel codes have also been developed in recent years but are only limited to short to moderate blocklengths \cite{Letizia:2021,Jiang:2019,Jamali:2022, Hebbar:2022}. Such black box approaches to communication system design require theoretical justification on the generalization properties of several advanced machine learning methods and DNN. Very few works have applied statistical learning tools to analyze the error rate performance and generalization properties of learning-based communication systems \cite{Angjelichinoski:2019,Weinberger:2022,Tsvieli:2022}.

In this paper, we investigate the fundamental limits of data-driven communication in the finite blocklength regime. We introduce the notion of \emph{learning-based channel codes} (see Definition \ref{definition:learnable_code}) and establish non-asymptotic performance bounds for this class of channel codes. These bounds are data-dependent and are proved by merging mathematical tools from finite blocklength information theory and statistical learning theory. The main contributions of this work are summarized as follows:
\begin{itemize}
    \item \textbf{Achievability:} We establish a data-dependent upper bound on the lowest achievable error probability of learning-based channel codes (see Theorem \ref{theorem:RCU_with_learning}).
    \item \textbf{Converse:} We prove a data-dependent upper bound on the codebook size that every learning-based channel code must satisfy (see Theorem \ref{theorem:metaconverse_learning}).
    \item \textbf{Asymptotic Expansion:} Under some conditions on the size of the training data, we characterize the behavior of the maximum achievable rate for sufficiently large blocklengths (see Theorem  \ref{theorem:normal_approx_achievability}). The upper and lower bounds established for the maximum achievable rate converge when the size of the training data is sufficiently large.
\end{itemize}

The rest of the paper is organized as follows: Section \ref{section:problem_formulation} formulates the problem setup and states the model assumptions. Section \ref{section:main_result} presents the non-asymptotic performance bounds of the learning-based channel codes. Section \ref{section:normal_approx} provides the asymptotic expansion of the bounds in Section \ref{section:main_result} to characterize the maximum achievable code rate of the learning-based channel codes. Finally, Section \ref{section:summary} concludes the paper. 
The reader is referred to the supplement file for the proofs of the main results.

\begin{figure*}[t]
    \centering
    \includegraphics[scale =1]{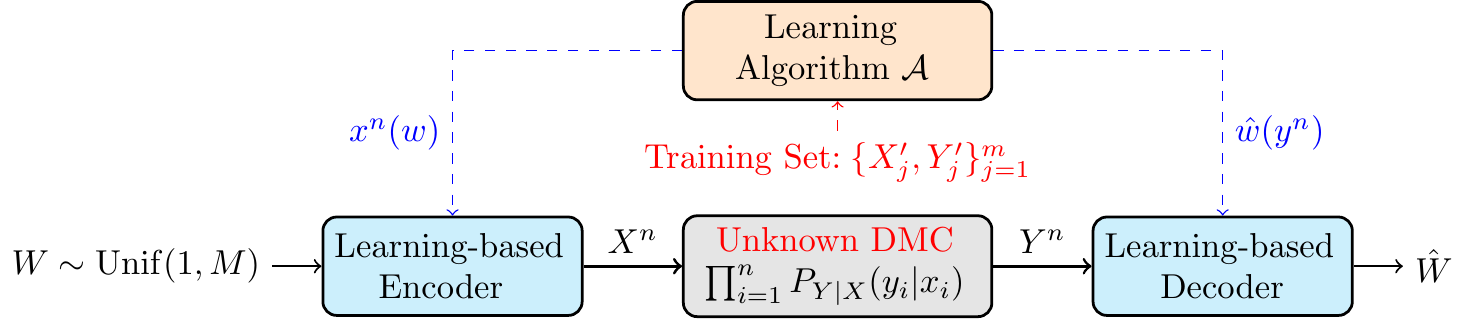}
    \caption{System model of the learning-based communication system.}
    \label{fig:learning2comm_problem_setup}
\end{figure*}

\emph{Notation:} All $\log(\cdot)$ terms in this paper are in base 2 unless specified otherwise. We use $\ln(\cdot)$ for the natural logarithm. When it is clear from the context, we use $P_{X}$ for the marginal distribution $P_{X}(x)$. We denote an $n$-fold product distribution of $P_{X}$ as $P_{X}^{n}$. \textcolor{black}{For the joint distribution of a length-$n$ sequence of random variables $X^n = (X_1,\cdots, X_n)$ (not necessarily a product distribution), we use $P_{X^n}$.} If $P$ is a probability measure and $E$ is an event, then $P[E] = \int_{E}dP$ is the probability of event $E$ under measure $P$. We denote the expectation operator by $\mathbb{E}[\cdot]$ and the variance operator by $\mathrm{Var}[\cdot]$. Equivalent notations can be used for the conditional and joint distributions. The function $\mathbbm{1}_{\{E\}}$ is the indicator function, which outputs 1 if $E$ is true and 0 otherwise. We use standard asymptotic notation in this paper: Suppose $f(n)$ and $g(n)$ are functions. We have $f(n) = \mathcal{O}(g(n))$ if there exists a constant $K$ such that $\lim_{n\rightarrow\infty}f(n)/g(n) < K$. Finally, let $f$ and $g$ be two quantities. We use the relationship $f \gtrsim g$ if there exists a constant $ K \geq 1$ such that $f \geq K\cdot g$.

%% file: problem_formulation.tex

We consider the point-to-point communication channel depicted in Figure \ref{fig:learning2comm_problem_setup}. The sender wishes to communicate a message $W$ to a receiver over a discrete memoryless channel (DMC) using $n$ channel uses. The message $W$ is drawn according to a uniform distribution over an index set $\{1,2,\cdots,M\}$. The DMC is the tuple $(\mathcal{X},P_{Y|X}(y|x),\mathcal{Y})$, where $\mathcal{X}$ is the input alphabet, $\mathcal{Y}$ is the output alphabet, and $P_{Y|X}(y|x)$ is the conditional probability mass function (pmf) representing the channel transition probabilities. We assume that the sizes of the input and output alphabets are finite, i.e. $|\mathcal{X}| < \infty$ and $|\mathcal{Y}| < \infty$. There is also a codebook $\mathcal{C}_{M, n}$ that contains $M$ length-$n$ codewords $\{c_1,\cdots,c_M\}$. This codebook is known by the encoder and decoder \emph{a priori}. The codebook $\mathcal{C}_{M, n}$ can be thought of as an $M\times n$ matrix where the $m$-th row contains the $m$-th codeword.

The channel transition probabilities of the DMC are not known to the sender and receiver. Instead, the sender and receiver are given a training set $\mathcal{D}_{m}$ which consists of $m$ channel input-output pairs $\{x_j',y_j'\}_{j = 1}^{m}$. Each channel input-output pair is drawn independently according to an unknown generating distribution $P_{Y|X}(y|x)\cdot (1/|\mathcal{X}|)$. Using the training set $\mathcal{D}_{m}$, the sender and receiver learn an encoding function $x^n(w)$ and a decoding function $\hat{w}(y^n)$, respectively, to communicate over the unknown channel. We emphasize that this is different from the problem setup of a compound channel (c.f. \cite[Chapter 7.2]{NetworkInfoTheory:2011}). The compound channel problem assumes that the channel law belongs to a set of conditional pmfs parameterized by a state $S\in\mathcal{S}$ and is fixed throughout transmission. This prior information is not available in our problem setup. In exchange for not knowing the set of conditional pmfs, a training set $\mathcal{D}_m$ is provided.



In the following, we formally define an $(M,n,\epsilon,\delta)$ learning-based channel code for a training set $\mathcal{D}_m$.


\begin{definition}\label{definition:learnable_code}
An $(M,n,\epsilon,\delta)$ learning-based channel code for a training set $\mathcal{D}_m  = \{x'_j,y_j'\}_{j = 1}^m$, where $(x_j',y_j')\sim P_{Y|X}(y|x)\frac{1}{|\mathcal{X}|}\;\forall (x,y)\in \mathcal{X}\times\mathcal{Y}$ consists of:
\begin{enumerate}
\item a message $W$ uniform over $[1:M]$,
\item a learning algorithm $\mathcal{A}: \mathcal{D}_{m} \mapsto (x^n(w),\hat{w}(y^n))$ that takes in the training set $\mathcal{D}_m$ and outputs an encoder function $x^n(w)$ and a decoder function $\hat{w}(y^n)$,
\item an encoder function $x^n(w): [1:M] \mapsto \mathcal{X}^n$ that assigns an $n$-length codeword to a message $w$, and
\item a decoder function $\hat{w}(y^n): \mathcal{Y}^n\mapsto [1:M]$ that maps the received sequence $y^n$ into a message $\hat{w}$. The decoder should satisfy the statistical reliability constraint
\begin{align}\label{eq:statistical_reliability}
    \mathbb{P}_{\mathcal{D}_m}\left\{\mathrm{P}_{\mathrm{e}}^{(M,n)} \leq \epsilon\right\} \geq 1-\delta,
\end{align}
where $\mathrm{P}_{\mathrm{e}}^{(M,n)}$ is the average error probability of the learned encoding and decoding function from the training set $\mathcal{D}_m$. This quantity is defined as
\begin{align}
    \mathrm{P}_{\mathrm{e}}^{(M,n)} \triangleq \frac{1}{M}\sum_{i = 1}^{M}\mathbb{P}\left(\hat{w}(Y^n) \neq W\;|\; x^n(W)\;\mathrm{ sent}\right).
\end{align}
\end{enumerate}
\end{definition}

There are several differences between the $(M,n,\epsilon,\delta)$ learning-based channel code in Definition \ref{definition:learnable_code} and the $(M,n,\epsilon')$ channel code commonly used in the finite blocklength literature (c.f. \cite{Polyanskiy:2010}). Firstly, a training set $\mathcal{D}_{m}$ and a learning algorithm $\mathcal{A}$ are introduced. The training set $\mathcal{D}_{m}$ contains partial knowledge of the unknown channel law and this knowledge is utilized by the learning algorithm $\mathcal{A}$ for the code design. Secondly, we introduce a confidence parameter $\delta$ which indicates how likely our channel code meet the error probability constraint $\mathrm{P}_{\mathrm{e}}^{(M,n)} \leq \epsilon$. This is necessary due to the stochastic nature of the training set $\mathcal{D}_{m}$. There may be a small chance that the training set $\mathcal{D}_{m}$ is uninformative (e.g. the training set contains $m$ repetitions of $x'$ and $y'$ even though $|\mathcal{X}| > 1$ and $|\mathcal{Y}| > 1$). Consequently, the learning algorithm might produce ``bad'' encoding and decoding functions such that $\mathrm{P}_{\mathrm{e}}^{(M,n)} > \epsilon$. The reliability constraint in \eqref{eq:statistical_reliability} resembles the criterion for the probably approximately correct (PAC) learning framework in the statistical learning theory literature.


While Definition \ref{definition:learnable_code} treats the average error probability as a random variable, the following lemma shows that $\mathrm{P}_{\mathrm{e}}^{(M,n)}$, averaged over all realizations of $\mathcal{D}_m$, is bounded above by the $\epsilon$ and $\delta$ parameters of the learning-based channel code. Ideally, we want both $\epsilon$ and $\delta$ to be as small as possible.
\begin{lemma}\label{lemma:learning_code_property}
For any $(M,n,\epsilon,\delta)$ learning-based channel code, we have $\mathbb{E}_{\mathcal{D}_m}\left[\mathrm{P}_{\mathrm{e}}^{(M,n)}\right] \leq \epsilon+\delta$.
\end{lemma}
\begin{proof}
    From \eqref{eq:statistical_reliability}, we have
     \begingroup
    \allowdisplaybreaks
    \begin{align*}
        \mathbb{E}[\mathrm{P}_{\mathrm{e}}^{(M,n)}] \leq &\mathbb{E}_{\mathcal{D}_m}\left[\epsilon \cdot\mathbbm{1}_{\{\mathrm{P}_{\mathrm{e}}^{(M,n)} \leq \epsilon\}} + 1\cdot \mathbbm{1}_{\{\mathrm{P}_{\mathrm{e}}^{(M,n)} > \epsilon\}}\right]\\
        = & \epsilon\cdot\mathbb{P}_{\mathcal{D}_m}\left\{\mathrm{P}_{\mathrm{e}}^{(M,n)} \leq \epsilon\right\}+1\cdot\mathbb{P}_{\mathcal{D}_m}\left\{\mathrm{P}_{\mathrm{e}}^{(M,n)} > \epsilon\right\}\\
        \leq & \epsilon\cdot 1 + 1\cdot\delta = \epsilon + \delta.
    \end{align*}
    \endgroup
\end{proof}

From this problem formulation we are now led to the following questions: 
\begin{itemize}
    \item \emph{Given a message index set of size $M$, $n$ channel uses, and a confidence parameter $\delta$, what is the smallest achievable probability of error $\epsilon^*$ that can be achieved by a learning-based channel code?} 
    \item \emph{Suppose we fix a target error probability $\epsilon$ and a confidence parameter $\delta$. Define the code rate $R$ as}
    \begin{align}
        R \triangleq \frac{\log M}{n} \;\mathrm{in}\;\mathrm{bits/channel\;use}.
    \end{align}
    \emph{For a fixed blocklength $n$, what is the highest code rate $R^*$ of the learning-based code such that the statistical reliability constraint is satisfied?} 
\end{itemize}

It is difficult to identify $\epsilon^*$ and $R^*$ in general, even if the channel is perfectly known by the encoder and decoder. To this end, we establish upper and lower bounds for the desired quantities. Since we do not have direct access to the full description of the channel, the bounds should depend on the training set $\mathcal{D}_{m}$ and not on the unknown channel $P_{Y|X}(y|x)$.  Such requirement is natural in many statistical learning problems so we borrow some tools from statistical learning theory to establish theoretical performance guarantees of learning-based channel codes in terms of $\mathcal{D}_m$.

%% file: results.tex

Before we state our main results, we first present two propositions that will play important roles in our derivations.

\begin{proposition} 
(\textbf{Change of measure inequality from the constrained representation of the total variation distance \cite[Lemma 4]{Ohnishi:2021}})\label{proposition:var_representation_TV} Let $\phi: \mathcal{H}\mapsto [0,1]$ be a real-valued function. Let $P$ and $Q$ denote the probability measures over the Borel $\sigma$-field on $\mathcal{H}$. Then,
\begin{align}
    \forall \text{$Q$ on $\mathcal{H}$}:\; \mathbb{E}_{Q}[\phi] \leq \mathrm{TV}(Q,P) + \mathbb{E}_{P}[\phi],
\end{align}
where $\mathrm{TV}(Q,P) \triangleq \sup_{E} |P[E] - Q[E]|$ is the total variation distance and the supremum is over all events $E$.
\end{proposition}

\begin{proposition}(\textbf{Bound on KL-divergence between empirical and true distribution\cite[Theorem 3]{Seldin:2009}}) \label{proposition:KL_bound} Let $X_1,\cdots, X_N$ be i.i.d. drawn according to
$P_X(x)$ and let $|X|$ be the cardinality of $X$. Denote by
$\hat{P}_X(x)$ the empirical distribution of $X_1,\cdots, X_N$. Then, with probability greater than $1-\delta$:
\begin{align}
    \mathrm{KL}\left(\hat{P}_X||P_{X}\right) \leq \frac{(|X|-1)\ln (N+1) - \ln \delta}{N},
\end{align}
where $\mathrm{KL}\left(\hat{P}_X||P_{X}\right)$ is the Kullback-Leibler (KL) divergence.
\end{proposition}

Proposition \ref{proposition:var_representation_TV} is used to change the measure of the expectation from the (unknown) true distribution $P_{X,Y}$ to the empirical distribution $\hat{P}_{X,Y}$. Proposition \ref{proposition:KL_bound} is combined with Bretagnolle-Huber Inequality \cite{Bretagnolle:1978} to obtain data-dependent bounds.

\subsection{Achievability Result}

We now state the achievability result of this paper. The achievability result is primarily based on the random coding union (RCU) bound established in \cite[Theorem 16]{Polyanskiy:2010}.
\begin{theorem}\label{theorem:RCU_with_learning}
    Fix a channel input distribution $P_{X}(x)$. Let $\mathcal{D}_{m}$ be the size-$m$ training set on which the encoder and decoder are trained. Then, there exists a $(2^{nR},n,\epsilon,\delta)$ learning-based code whose error probability is upper bounded as
\begin{align}\label{eq:RCU_with_learning}
        \epsilon^*\leq& \min_{\substack{n_0 \leq n\\ n_0\in\mathbb{Z}^+}}\Bigg\{\mathbb{E}_{P_{X}^{n_0}\hat{P}_{Y|X}^{n_0}}\left[\min\left\{1,\frac{L(\lceil2^{\frac{nR}{L}}\rceil^L-1)}{2^{i_{\hat{P}}\left(X^{n_0},Y^{n_0}\right)}}\right\}\right]\nonumber\\
        &+\sqrt{1-\exp\left(-\frac{n_0\left((|\mathcal{X}||\mathcal{Y}|-1)\ln (m+1) - \ln\delta\right)}{m}\right)}\Bigg\},
    \end{align}
    where $L = \lfloor n/n_0\rfloor$, $\hat{P}_{Y|X}$ is the estimate of the unknown channel law $P_{Y|X}$ based on $\mathcal{D}_m$ which can be expressed as
    \begin{align*}
    \hat{P}_{Y|X}(y|x) 
    = \frac{\sum_{\mathcal{D}_m} \mathbbm{1}_{\left\{ (x,y) = (x_i',y_i')\right\}} }{\sum_{\mathcal{D}_m}\mathbbm{1}_{\{x = x_i'\}}},\;\forall (x,y)\in \mathcal{X}\times\mathcal{Y},
\end{align*}
    and $i_{\hat{P}}(x,y)$ is the mutual information density with respect to the measure $\hat{P} = \hat{P}_{Y|X}P_{X}$ and can be written as
    \begin{align}\label{eq:information_density}
        i_{\hat{P}}(x,y) \triangleq \log \frac{\hat{P}_{Y|X}(y|x)}{\mathbb{E}[\hat{P}_{Y|X}(y|X)]}.
    \end{align}
\end{theorem}
\begin{proof}
See Appendix \ref{proof:achievability}.
\end{proof}
The idea behind Theorem \ref{theorem:RCU_with_learning} is to design a mini codebook $\mathcal{C}_{M_0,n_0}$ whose $L$-th order extension satisfies $M_0^L \geq 2^{nR}$ and $n_0L \leq n$. We then analyze this extended codebook. The receiver implements a block-wise decoder for the mini codebook and applies it $L$ times to the the length-$n_0L$ received sequence to recover the message $w$. \textcolor{black}{The key steps in upper bounding the error probability closely follow the proof of the RCU bound in \cite{Polyanskiy:2010}. Then, Propositions \ref{proposition:var_representation_TV} and \ref{proposition:KL_bound} are used to remove the dependence of the bound on the unknown DMC $P_{Y|X}$.}

The result has a nice intuition. The first term of \eqref{eq:RCU_with_learning} is the RCU bound with respect to the empirical measure $\hat{P}_{Y|X}$ and the learned decoder. Meanwhile, the second term accounts for the ``distance'' between our empirical product distribution and the true product distribution. Increasing $n_0$ reduces the first term since we are encoding using larger blocklengths. However, this first term could be far from the error incurred by the mismatched decoder on the unknown true channel. This discrepancy is captured by the second term, which grows with $n_0$. Hence, the parameter $n_0$ should be chosen carefully.

\emph{Remark:} Theorem \ref{theorem:RCU_with_learning} does not necessarily imply that it is better to encode the data in concatenated short packets instead of a single large block. Rather, when $P_{Y|X}$ is not available in our analysis, we can upper bound the error using the empirical product measure $\hat{P}_{Y|X}^{n_0}$ obtained from $\mathcal{D}_m$ but with some penalty term that grows with $n_0$ due to the lack of information about the channel. This penalty term might not be tight in general.



\subsection{Converse Result}
To establish the converse result, we present relevant concepts in binary hypothesis testing. 

The main objective of binary hypothesis testing is to determine if a sequence $X^n$ is generated according to a distribution \textcolor{black}{$P_{X^n}$} or an alternate distribution \textcolor{black}{$Q_{X^n}$}. A test $Z$ is designed which outputs 0 when it chooses \textcolor{black}{$P_{X^n}$} and 1 when it chooses \textcolor{black}{$Q_{X^n}$}. We focus on randomized tests defined by the stochastic kernel $P_{Z|X^n}(z|x^n): \mathcal{X}^n\mapsto \{0,1\}$. The following definition plays an important role in the binary hypothesis testing problem:
\begin{definition}(\textbf{Neyman-Pearson Region and Beta Function})\label{definition:neyman_pearson_func}
Suppose we define $P[Z = 0] = \int P_{Z|X^n}(0|x^n)dP_{X^n}$ 
 as the probability of success given $X^n\sim P_{X^n}$ and $Q[Z = 0] = \int P_{Z|X^n}(0|x^n)dQ_{X^n}$ as the probability of error given $X^n\sim Q_{X^n}$. Then, the Neyman-Pearson region is the set of all achievable points for all randomized tests:
\begin{align}
    \mathcal{R}(P,Q) \triangleq \left\{(P[Z=0],Q[Z=0]),\;\forall\;P_{Z|X^n}\right\}.
\end{align}
The Neyman-Pearson beta function is the lower boundary of $\mathcal{R}(P,Q)$ and is defined as
\begin{align}
 \textcolor{black}{\beta_{\alpha}(P_{X^n},Q_{X^n})} = \inf_{P_{Z|X^n}:P[Z = 0] \geq \alpha} Q[Z = 0],
\end{align}
where $\alpha \in [0,1]$.
\end{definition}
We now state the converse result of this paper. The converse result is primarily based on the metaconverse theorem established in \cite[Theorem 27]{Polyanskiy:2010}. \textcolor{black}{Similar to Theorem \ref{theorem:RCU_with_learning}, the dependence of the converse bound on the unknown $P_{Y|X}$ is removed by applying Propositions \ref{proposition:var_representation_TV} and \ref{proposition:KL_bound}. }

\begin{theorem}\label{theorem:metaconverse_learning}
    Fix an arbitrary channel output distribution \textcolor{black}{$Q_{Y^n}(y^n)$}. Let $\mathcal{D}_m$ be the size-$m$ training set on which the encoder and decoder are trained and $\hat{P}_{Y|X}$ be the empirical distribution of the channel law based on $\mathcal{D}_{m}$. Every $(2^{nR},n,\epsilon,\delta)$ learning-based code must satisfy
    \textcolor{black}{
\begin{align}\label{eq:metaconverse_learning}
        2^{nR} \leq \sup_{P_{X^n}}\frac{1}{\beta_{\max\{0,1-\epsilon-\kappa\}}\left(\hat{P}_{Y|X}^{n}P_{X^n},Q_{Y^n}P_{X^n}\right)},
    \end{align}
    }
    where $\kappa = \sqrt{1-\exp\left(-\frac{n\left((|\mathcal{X}||\mathcal{Y}|-1)\ln (m+1) - \ln \delta\right)}{m}\right)}$.
\end{theorem}
\begin{proof}
See Appendix \ref{proof:converse}.
\end{proof}


The parameter $\kappa$ arises from the fact that we evaluate the bound using $\hat{P}_{Y|X}$ rather than $P_{Y|X}$. The converse result assumes that the average probability of correct decoding, denoted $\mathrm{P}_{\mathrm{c}}$, of the true channel $P_{Y|X}$ is better than that of our empirical estimate by at least $\kappa$. Hence, we reduce the requirement on $\mathrm{P}_{\mathrm{c}}$ by $\kappa$. Note, however, that $m$ should be sufficiently large for the bound to be non-vacuous, i.e. we want $\kappa + \epsilon < 1$ otherwise we get the trivial bound $2^{nR} \leq \infty$. Lastly, we emphasize that the presence of $\hat{P}_{Y|X}$ in \eqref{eq:metaconverse_learning} is for evaluating the bound using the training data $\mathcal{D}_m$, i.e. a data-dependent bound.

%% file: normal_approx.tex
The bounds in Theorem \ref{theorem:RCU_with_learning} and Theorem \ref{theorem:metaconverse_learning} are not in closed-form. In this section, we perform asymptotic expansions of the bounds. More specifically, we use normal approximation to characterize bounds for $R^*(n,m,\epsilon,\delta)$, the maximum achievable rate of a learning-based channel code for a fixed blocklength $n$, fixed training set size $m$, target error probability $\epsilon$, and confidence parameter $\delta$. Such expansions give insights on how $R^*(n,m,\epsilon,\delta)$ converges to the capacity as the blocklength increases. The following variant of central limit theorem (CLT) will be critical in the analysis.
\begin{theorem}(\textbf{Berry-Esseen CLT \cite[Theorem 2, Chapter
XVI.5]{Feller:1957}})\label{theorem:berry_esseen} Let $\{X_{k}\}_{k=1}^{n}$ be independent random variables with means $\mu_k = \mathbb{E}[X_{k}]$, variances $\sigma_{k}^2 = \mathrm{Var}\left[X_{k}\right]$, and third absolute central moments $\theta_k = \mathbb{E}[|X_{k}-\mu_k|^3]$. Let $\sigma^2 = \sum_{k=1}^{n}\sigma_k^2$ and $\Theta = \sum_{k=1}^{n}\theta_k$. Then, for any $t \in \mathbb{R}$, we have
\begin{align}
    \left|\mathbb{P}\left\{\frac{\sum_{k =1}^n\left(X_{k} - \mu_{k}\right)}{\sigma}\geq t\right\} - \mathcal{Q}\left(t\right)\right|\leq \frac{6\Theta}{\sigma^{3}},
\end{align}
where $\mathcal{Q}(\cdot)$ is the tail probability of the standard Gaussian distribution.
\end{theorem}
When all $X_{k}$'s are also identically distributed and we set $\theta_k = \theta_0$ and $\sigma_k^2 = \sigma_0^2$ for all $k$, then the RHS of Theorem \ref{theorem:berry_esseen} becomes $\frac{6\theta_0}{\sigma_0^3\sqrt{n}}$, i.e. the error term decays as $1/\sqrt{n}$.

We now state the normal approximation result for the maximum code rate of learning-based channel codes. 

\begin{theorem}\label{theorem:normal_approx_achievability}
    
   Assume that $m$ and $n$ are sufficiently large, and it holds that 
\begin{align}\label{eq:blocklength_condition}
    n \leq \sqrt{\frac{m}{(|\mathcal{X}||\mathcal{Y}|-1)\ln (m+1) - \ln(\delta)}}.
   \end{align}
 If $\hat{P}_{Y|X}$ is not an exotic DMC\footnote{The precise definition of an exotic channel is provided in \cite[Section IV]{Polyanskiy:2010}. As emphasized in the example given in \cite[Appendix H]{Polyanskiy:2010}, conditions for exotic channels are difficult to satisfy.}, then the maximum code rate $R^*$ of a $(2^{nR},n,\epsilon,\delta)$ learning-based channel code is 
\begin{align}\label{eq:rate_normal}
    R^* = & C_{\hat{P}} - \sqrt{\frac{V_{\hat{P}}^{\epsilon}}{n}}\mathcal{Q}^{-1}\left(\epsilon\right) + \mathcal{O}\left(\frac{\log n}{n}\right),
    \end{align}
where $C_{\hat{P}} = \max_{P_{X}}I(P_{X},\hat{P}_{Y|X})$ is the channel capacity of $\hat{P}_{Y|X}$ and $V_{\hat{P}}^{\epsilon}$ is the channel dispersion of $\hat{P}_{Y|X}$ which can be written as
\begin{align}\label{eq:channel_dispersion}
    V_{\hat{P}}^{\epsilon} \triangleq \begin{cases}\max_{P_{X}\in\mathcal{\hat{P}}_X^*}\mathrm{Var}[i_{\hat{P}}(X,Y)|X],\quad \epsilon > \frac{1}{2}\\
    \;\min_{P_{X}\in\mathcal{\hat{P}}_X^*}\mathrm{Var}[i_{\hat{P}}(X,Y)|X],\quad \epsilon < \frac{1}{2}
    \end{cases}.
\end{align}
The optimization of the channel dispersion is over $\mathcal{\hat{P}}_X^*$, the set of all input distributions $P_{X}$ that achieve the capacity $C_{\hat{P}}$.
\end{theorem}
\begin{proof}
    See Appendix \ref{proof:normal_approx_achievability}.
\end{proof}

The quantities $C_{\hat{P}}$ and $V_{\hat{P}}^{\epsilon}$ involve optimization of the input distribution $P_{X}$. The learning algorithm first solves for $\mathcal{\hat{P}}_{X}^*$, the set of all capacity-achieving input distributions\footnote{There may be multiple capacity-achieving input distributions that lead to the unique capacity-achieving output distribution.} for the empirical channel $\hat{P}_{Y|X}$. The learning algorithm then picks the distribution from the set $\mathcal{P}_{X}^*$ which maximizes (resp. minimizes) $\mathrm{Var}[i(X,Y)|X]$ when $\epsilon > \frac{1}{2}$ (resp. when $\epsilon < \frac{1}{2}$). The solution, which we denote as $\hat{P}_{X}^*$, is then used to construct the codebook. 

The result in Theorem \ref{theorem:normal_approx_achievability} has a striking resemblance with the normal approximation result in \cite[Section IV]{Polyanskiy:2010}. The difference is that we replace the channel capacity $C$ and the channel dispersion $V^{\epsilon}$ of the channel $P_{Y|X}$ with those of the empirical channel $P_{Y|X}$. It is known that the empirical DMC $\hat{P}_{Y|X}$ converges to $P_{Y|X}$ with probability 1 as $m$ grows to infinity (see Glivenko–Cantelli theorem in \cite{Tucker:1959}). However, we note that $C_{\hat{P}}$ and $V_{\hat{P}}^{\epsilon}$ do not necessarily converge to the channel capacity and the channel dispersion of the true channel $P_{Y|X}$, respectively. To the best of the authors' knowledge, there are no existing results showing that the channel capacity and channel dispersion of the empirical channel law converge to those of the true channel law as the sample size $m$ grows to infinity.

It is also crucial to point out that Theorem \ref{theorem:normal_approx_achievability} holds when $m \gg n$ (as observed in condition \eqref{eq:blocklength_condition}). In fact, $m$ is required to grow superlinearly with respect to $n$. Hence, Theorem \ref{theorem:normal_approx_achievability} cannot be adapted to scenarios where the learning phase is done during transmission since there are not enough channel uses to facilitate both learning phase and communication phase. The super linear growth requirement in $m$ is an artefact of our high probability bound for $\mathrm{TV}(\hat{P}_{Y|X}^n,P_{Y|X}^n)$. It remains unclear if this scaling of the training set size can be improved if we insist on obtaining bounds that depend entirely on the training data and not on the unknown channel distribution. Moreover, it can be observed that the maximum code rate in \eqref{eq:rate_normal} does not depend on the confidence parameter $\delta$, implying that this rate expression holds for any $\delta$. This is because we are enforcing condition \eqref{eq:blocklength_condition} to hold in Theorem \ref{theorem:normal_approx_achievability}. Hence, $m$ is large enough to accommodate the high probability requirement.

One major disadvantage of Theorem \ref{theorem:normal_approx_achievability} is the stringent requirement on the training size. The following corollary gives a lower bound on the maximum code rate of learning-based channel codes when condition \eqref{eq:blocklength_condition} is not met.
\begin{corollary}\label{corollary:normal_approx_general}  
    Assume $m$ and an integer $n_0 < n$ are sufficiently large, and it holds that
    \begin{align}\label{eq:blocklength_condition2}
        n_0 \leq \sqrt{\frac{m}{(|\mathcal{X}||\mathcal{Y}|-1)\ln (m+1) - \ln(\delta)}}.
    \end{align}
    If condition \eqref{eq:blocklength_condition} is not met and $\hat{P}_{Y|X}$ is not an exotic DMC,
   then there exists a $(2^{nR},n,\epsilon,\delta)$ learning-based channel code for which the code rate is
    \begin{align}\label{eq:rate_normal_n0}
    R_{\mathrm{achievable}}^* =& \frac{n_0 C_{\hat{P}}}{n} - \sqrt{\frac{n_0V_{\hat{P}}^{\epsilon}}{n^2}}\mathcal{Q}^{-1}\left(\epsilon\right)+ \mathcal{O}\left(\frac{1}{n}\right).
\end{align}
\end{corollary}
\begin{proof}
   See Appendix \ref{proof:normal_approx_general}.
\end{proof}
When condition \eqref{eq:blocklength_condition} is not met, we can instead look for an $n_0 < n$ that satisfies \eqref{eq:blocklength_condition2}. If this $n_0$ is sufficiently large, then we get the achievability result in Corollary \ref{corollary:normal_approx_general}. The factors $\frac{n_0}{n}$ and $\sqrt{\frac{n_0}{n}}$ are multiplied to the empirical capacity $C_{\hat{P}}$ and empirical channel dispersion $V_{\hat{P}}^{\epsilon}$, respectively.


As a final remark, Theorem \ref{theorem:normal_approx_achievability} and Corollary \ref{corollary:normal_approx_general} did not specify exact values of what \emph{sufficiently large} $n$ and \emph{sufficiently large} $n_0$ should be. Based on the proofs, we must have $n \gtrsim \epsilon^{-2}$ and $n_0 \gtrsim \epsilon^{-2}$ for Theorem \ref{theorem:normal_approx_achievability} and Corollary \ref{corollary:normal_approx_general}, respectively.

%% file: summary.tex
In this work, we introduced a class of channel codes, which we coined learning-based channel codes, to facilitate a data-driven approach in designing channel codes for model-deficit communication problems. This class of channel codes includes a learning algorithm which maps the training set to a channel encoder-decoder pair. The channel encoder-decoder pair should satisfy the statistical reliability constraint in \eqref{eq:statistical_reliability}. We established achievability and converse bounds for the maximum achievable rate $R^*$ of the learning-based channel codes for given blocklength $n$, training size $m$, target error probability $\epsilon$, and confidence parameter $\delta$ over a point-to-point DMC with unknown channel law. The bounds are data-dependent and do not require complete information of the underlying channel law to evaluate. For sufficiently large $m$, the established bounds converge to the maximum achievable rate of the empirical channel law. 

An important direction for future research is to establish tighter bounds for training set with small to moderate size. Currently, our asymptotic result requires the size of the training set to be much larger than the blocklength, which can be impractical. Moreover, the current result only holds for the discrete memoryless case. It is also of interest to extend the result to channels with continuous input and output alphabets as well as data-driven communication systems with feedback. Lastly, more refined asymptotic results are possible by using saddlepoint methods \cite{Polyanskiy:2012,Font-Segura:2018}.

%% file: acknowledgement.tex
This work has received funding from the Australian Research Council under project DE210101497. N.I. Bernardo acknowledges the Melbourne Research Scholarship of the University of Melbourne and the Department of Science and Technology-Engineering Research and Development for Technology (DOST-ERDT) Faculty Development Fund of the Republic of the Philippines for sponsoring his doctoral studies.

%% file: appendix_A.tex
The proof modifies the RCU bound to incorporate the learning aspect of the communication system.

\textbf{Codebook Generation:} Pick an integer $n_0 \in [1,n]$ and let $L = \left\lfloor n/n_0\right\rfloor$. Let $M_0 = \left\lceil2^{bR}\right\rceil$, where $b = n/L$. With this choice of $b$, it is easy to verify that $M_0^L \geq M = 2^{nR}$. We first construct a mini codebook $\mathcal{C}_{M_0,n_0}$. We randomly generate $M_0$ $n_0$-length sub-codeword according to the distribution $P_{X}^{n_0} = \prod_{i = 1}^{n_0}P_{X}(x_i)$. Each sub-codeword is mapped to an index set $\{1,\cdots,M_0\}$ to form the mini codebook $\mathcal{C}_{M_0,n_0}$. The codebook $\mathcal{C}_{M_0^L,n_0L}$ is then the $L$-th order extension of $\mathcal{C}_{M_0,n_0}$. Since $n_0 L \leq n$, the code rate of $\mathcal{C}_{M_0^L,n_0L}$ is higher than $\mathcal{C}_{M,n}$. In the subsequent, we analyze this code.

\textbf{Encoder Design:} Apply the mapping $ w \mapsto X^{n_0L}(w)$, where $X^{n_0L}(w)$ is the $w$-th codeword of $\mathcal{C}_{M_0^L,n_0L}$.

\textbf{Decoder Design:} Using $\mathcal{D}_{m}$, the decoder learns an empirical version of the maximum likelihood (ML) decoding rule for the mini codebook, denoted $\hat{P}_{\mathrm{ML}}^{n_0}(\cdot,*)$. This decoding rule is applied $L$ times to decode the received sequence $y^n$. Denote $x^{(l)}$ to be the $n_0$-length subsequence at the $l$-th extension, i.e. $x^{(l)} = [x_{1 + n_0(l-1)},\cdots, x_{n_0 + n_0(l-1)}]$, and use the same notation for $y^{(l)}$. The learned decoding rule works as follows:
\begin{itemize}
\item\textbf{Step 1:} Apply $\hat{P}_{\mathrm{ML}}^{n_0}(\cdot,*)$ to each sub block:
\begin{align}\label{eq:empirical_ML}
    \hat{x}^{(l)} =& 
    \;\;
    \hat{P}_{\mathrm{ML}}^{n_0}(x^{(l)},y^{(l)})\nonumber\\
    =& \underset{x^{(l)}}{\arg\max}\; \prod_{i = 1}^{n_0} \hat{P}_{Y|X}(y_{i+n_0(l-1)}|x_{i+n_0(l-1)})
\end{align}
for all $l\in[1,L]$, where
\begin{align*}
    \hat{P}_{Y|X}(y|x) 
    = \frac{\sum_{\mathcal{D}_m} \mathbbm{1}_{\left\{ (x,y) = (x_i',y_i')\right\}} }{\sum_{\mathcal{D}_m}\mathbbm{1}_{\{x = x_i'\}}},\;\forall (x,y)\in \mathcal{X}\times\mathcal{Y}. 
\end{align*}
\item \textbf{Step 2:} Map $[\hat{x}^{(1)},\cdots,\hat{x}^{(L)}]$ to $\hat{w}$. If the decoded sequence does not appear in $\mathcal{C}_{M_0^L,n_0L}$, declare an error.
\end{itemize}
\textbf{Error Probability Analysis:} Let $\lambda_j = \mathbb{P}\{\mathrm{error}|X^{n_0L} = c_j\}$, where $c_j$ is the $j$-th codeword of $\mathcal{C}_{M_0^L,n_0L}$. The average probability, as a function of the codebook, is
\begin{align*}
    \epsilon(c_1,\cdots,c_{M_0^L}) = \frac{1}{M_0^L}\sum_{j = 1}^{M_0^L}\lambda_j
\end{align*}
for which, by symmetry, we get
\begin{align*}
    \epsilon^* \leq \epsilon = \mathbb{E}[\epsilon(C_1,\cdots,C_{M_0^L})] = \mathbb{E}[\lambda_1].
\end{align*}
We now average $\lambda_1$ over the random choice of codebook. Let the sequence \[X^{n_0L} = [X^{(1)},\cdots,X^{(l)},\cdots,X^{(L)}]\] be the transmitted codeword $C_1$ and \[\tilde{X}^{n_0L}= [\tilde{X}^{(1)},\cdots,\tilde{X}^{(l)},\cdots,\tilde{X}^{(L)}]\] be another codeword from the extended codebook $\mathcal{C}_{M_0^L,n_0L}$. Then, we can upper bound the error probability as 
\begin{align*}
    &\mathbb{E}[\lambda_1]\\
    &\leq  \mathbb{P}\left[\bigcup_{t = 2}^{M_0^L}\left\{\text{Decoder outputs $\tilde{X}^{n0L} = C_t$\;$,$\;$X^{n_0L} = C_1$}\right\}\right]\\
    &\leq \mathbb{E}\left[ \min\left\{1,\sum_{t = 2}^{M_0^L}\mathbb{P}\left[\text{Decoder outputs $\tilde{X}^{n_0L}\big|X^{n_0L}$,$Y^{n_0L}$}\right]\right\}\right],
\end{align*}
where the second inequality comes from applying the union bound and conditioning on $(X^{n_0L},Y^{n_0L})$. We use $\min\{1,x\}$ to exclude the values that exceed 1.  Next, let 
\begin{align*}\lambda^{(l)} = \mathbf{P}\left\{\text{$\hat{P}_{\mathrm{ML}}^{n_0}(\cdot,*)$ picks $\tilde{X}^{(l)}$\;\bigg| $X^{(l)}$ sent, $Y^{(l)}$ received} \right\}.\end{align*}
It should be noted that $\lambda^{(1)} = \cdots = \lambda^{(L)}$ since the mini codebooks are identically distributed. Then, we can bound the probability term as follows:
\begin{align*}
   \mathbb{P}\left[\text{Decoder outputs $\tilde{X}^{n_0L}\big|X^{n_0L}$,$Y^{n_0L}$}\right] =& 1 - (1-\lambda^{(1)})^L\\
   \leq & 1 - (1-L\lambda^{(1)})\\
   = & L\cdot \lambda^{(1)},
\end{align*}
where second line follows from the Bernoulli inequality. By also noting that the $M_0^L-1$ probability terms are identical, we get
\begin{align*}
     \mathbb{E}[\lambda_1]\leq & \mathbb{E}_{P_{Y|X}^{n_0}P_{X}^{n_0}}\left[\min\left\{1,(M_0^L-1)L\lambda^{(1)}\right\}\right].
\end{align*}
This is equivalent to \eqref{eq:derivation1} when $\lambda^{(1)}$ is expanded. We further relax equation \eqref{eq:derivation1} to get \eqref{eq:derivation2}. By Markov's Inequality, the probability term inside the expectation becomes
\begin{align*}
\mathbb{P}\left(\hat{P}_{Y|X}^{n_0}(y^{n_0} |\tilde{X}^{n_0}) \geq  \hat{P}_{Y|X}^{n_0}(y^{n_0}|x^{n_0})\right) & \leq \frac{\mathbb{E}\left[\hat{P}_{Y|X}^{n_0}(y^{n_0}|\tilde{X}^{n_0})\right]}{\hat{P}_{Y|X}^{n_0}\left(y^{n_0}|x^{n_0}\right)}\\
& = 2^{-i_{\hat{P}}\left(x^{n_0},y^{n_0}\right)}.
\end{align*}

To get equation \eqref{eq:derivation3}, we apply Proposition \ref{proposition:var_representation_TV} to change the measure of the expectation from the unknown DMC $P_{Y|X}$ to the empirical channel $\hat{P}_{Y|X}$. The second term in \eqref{eq:derivation3} can be bounded using the Bretagnolle-Huber identity and the tensorization property of the KL divergence. This gives us \eqref{eq:derivation4}. The proof is completed by applying Proposition \ref{proposition:KL_bound} and optimizing the integer parameter $n_0$.

\begin{figure*}[!t]
\begin{align}
    \mathbb{E}[\lambda_1]
    \leq & \mathbb{E}_{P_{Y|X}^{n_0}P_{X}^{n_0}}\left[\min\left\{1,(M_0^L-1)L\mathbb{P}\left\{\hat{P}_{Y|X}^{n_0}\left(Y^{n_0}|\tilde{X}^{n_0}\right) \geq \hat{P}_{Y|X}^{n_0}\left(Y^{n_0}|X^{n_0}\right)\bigg| X^{n_0},Y^{n_0}\right\}\right\}\right]\label{eq:derivation1}\\
    \leq & \mathbb{E}_{P_{Y|X}^{n_0}P_{X}^{n_0}}\left[\min\left\{1,(M_0^L-1)L\cdot2^{-i_{\hat{P}}(X^{n_0},Y^{n_0})}\right\}\right]\label{eq:derivation2}\\
    \leq & \mathbb{E}_{\hat{P}_{Y|X}^{n_0}P_{X}^{n_0}}\left[\min\left\{1,(M_0^L-1)L\cdot2^{-i_{\hat{P}}(X^{n_0},Y^{n_0})}\right\}\right] +\mathrm{TV}\left(\hat{P}_{Y|X}^{n_0}P_{X}^{n_0},P_{Y|X}^{n_0}P_{X}^{n_0}\right)\label{eq:derivation3}\\
     \leq & \mathbb{E}_{\hat{P}_{Y|X}^{n_0}P_{X}^{n_0}}\left[\min\left\{1,(M_0^L-1)L\cdot2^{-i_{\hat{P}}(X^{n_0},Y^{n_0})}\right\}\right] +\sqrt{1-e^{-n_0\mathrm{KL}\left(\hat{P}_{Y|X}P_{X}||P_{Y|X}P_{X}\right)}}\label{eq:derivation4}
\end{align}
\hrulefill
\end{figure*}

%% file: appendix_B.tex
Our starting point is the metaconverse bound \cite[Theorem 27]{Polyanskiy:2010}. Fix an arbitrary distribution \textcolor{black}{$Q_{Y^n}$}. Then, every $(2^{nR},n,\epsilon)$ code for a given channel $P_{Y|X}$ should satisfy
\begin{align}\label{eq:metaconverse}
    2^{nR} \leq \sup_{P_{X^n}}\frac{1}{\beta_{1-\epsilon}\left(P_{Y|X}^{n}P_{X^n},Q_{Y^n}P_{X^n}\right)}.
\end{align}
Since the channel $P_{Y|X}$ is not known in the problem setup, we tweak the metaconverse bound to change the dependence on the true channel $P_{Y|X}$ to the empirical distribution $\hat{P}_{Y|X}$.

Let $t = \mathrm{TV}\left(P_{Y|X}^n\textcolor{black}{P_{X^n}},\hat{P}_{Y|X}^n\textcolor{black}{P_{X^n}}\right)$. We need to show that the RHS of \eqref{eq:metaconverse_learning} is larger than the RHS of \eqref{eq:metaconverse}. From Definition \ref{definition:neyman_pearson_func}, we have
\textcolor{black}{
\begin{align*}
    &\beta_{1-\epsilon}\left(P_{Y|X}^nP_{X^n}, Q_{Y^n}P_{X^n}\right) \\
    &\qquad= \inf_{P_{Y|X}^nP_{X^n}[Z = 0]\geq 1 - \epsilon} Q_{Y^n}P_{X^n}[Z = 0]\\
    &\qquad\geq \inf_{\hat{P}_{Y|X}^nP_{X^n}[Z = 0]+t\geq 1 - \epsilon} Q_{Y^n}P_{X^n}[Z = 0]\\
    &\qquad= \beta_{1-\epsilon-t}\left(\hat{P}_{Y|X}^nP_{X^n}, Q_{Y^n}P_{X^n}\right),
\end{align*}
}
where the infimum is over all tests $P_{Z|X^n}$ that satisfy the constraint. The inequality in the second line comes from Proposition \ref{proposition:var_representation_TV}. The inequality holds since, for a fixed objective function, the solution to a minimization problem does not increase when the search space is widened. Next, we establish a high probability bound on $t$. \textcolor{black}{The following inequality holds:}
\begin{align*}
    t\leq& \sqrt{1-e^{-\mathrm{KL}(\hat{P}_{Y|X}^nP_{X^n}||P_{Y|X}^nP_{X^n})}}.
\end{align*}
This follows from the Bretagnolle-Huber inequality \cite{Bretagnolle:1978}. \textcolor{black}{For any arbitrary DMC $\hat{P}_{Y|X}$, we can limit $P_{X^n}$ to be a distribution of an exchangeable sequence \cite{Polyanskiy:2012}. Hence, all $P_{X_i}$'s are identically distributed but are not necessarily independent. Consequently, our bound for $t$ becomes
\begin{align*}
    t\leq& \sqrt{1-e^{-\mathrm{KL}(\prod_{i=1}^{n}\hat{P}_{Y_i|X_i}P_{X_i|X^{i-1}}||\prod_{i=1}^{n}P_{Y_i|X_i}P_{X_i|X^{i-1}})}}\\
    =&\sqrt{1-e^{-\sum_{i=1}^n\mathrm{KL}(\hat{P}_{Y_iX_i|X^{i-1}}||P_{Y_iX_i|X^{i-1}})}}.
\end{align*}
The first line follows from $P_{X^n} = \prod_{i=1}^n P_{X_i|X^{i-1}}$ and rearranging the terms. The second line follows because
\begin{align*}
    &\mathrm{KL}(\prod_{i=1}^{n}\hat{P}_{Y_i|X_i}P_{X_i|X^{i-1}}||\prod_{i=1}^{n}P_{Y_i|X_i}P_{X_i|X^{i-1}})\\
    &=\mathrm{KL}(\prod_{i=1}^{n}\hat{P}_{Y_i|X_i}||\prod_{i=1}^{n}P_{Y_i|X_i}| P_{X^n})\\
    &=\mathbb{E}_{X^n\sim P_{X^n}}\left[\mathrm{KL}(\prod_{i=1}^{n}\hat{P}_{Y_i|X_i=x_i}||\prod_{i=1}^{n}P_{Y_i|X_i=x_i})\right]\\
    &=\sum_{i=1}^n\mathbb{E}_{X^n\sim P_{X^n}}\left[\mathrm{KL}(\hat{P}_{Y_i|X_i=x_i}||P_{Y_i|X_i=x_i})\right]\\
    &=\sum_{i=1}^n\mathrm{KL}(\hat{P}_{Y_i,X_i|X^{i-1}}||P_{Y_i,X_i|X^{i-1}}).
\end{align*}
Let $|\mathcal{X}_i\times \mathcal{Y}|$ be the number of possible values of $X_i\times Y$ when conditioned on the sequence $X^{i-1}$. Then, by Proposition \ref{proposition:KL_bound}, the bound 
\begin{align*}
     t\leq& \sqrt{1-\exp\left(-\sum_{i=1}^n\frac{\left((|\mathcal{X}_i \times \mathcal{Y}|-1)\ln (m+1) - \ln \delta\right)}{m}\right)}\\
     \leq& \sqrt{1-\exp\left(-\frac{n\left((|\mathcal{X}||\mathcal{Y}|-1)\ln (m+1) - \ln \delta\right)}{m}\right)} = \kappa
\end{align*}
holds with probability at least $1-\delta$. The second inequality follows because $P_{Y_i X_i}$ have the same number of possible values for all $i$ (due to exchangeability of the input and property of the DMC). This number is upper bounded by $|\mathcal{X}||\mathcal{Y}|$. Conditioning on the sequence $X^{i-1}$ will not increase this number, i.e. $|\mathcal{X}_i\times\mathcal{Y}|\leq |\mathcal{X}||\mathcal{Y}|$.} It is easy to verify that \eqref{eq:statistical_reliability} is satisfied. Note that the data-dependent bound is agnostic of the learning algorithm $\mathcal{A}$ used. Hence,  \eqref{eq:metaconverse_learning} holds for all learning algorithm whose output encoding and decoding functions satisfy \eqref{eq:statistical_reliability}. Noting that $\beta_{\alpha}(P,Q)$ is non-decreasing in $\alpha$ and that $\alpha\in [0,1]$ completes the proof.

%% file: appendix_C.tex
We start with the achievability part. Introduce a uniform random variable $U \sim \mathrm{Unif}(0,1)$. Then for any nonnegative $Z$, we have $\mathbb{E}[\min\{1,Z\}] = \mathbb{P}(Z \geq U)$. Note also that $i_{\hat{P}}(x^{n_0},y^{n_0}) = \sum_{i = 1}^{n_0}i_{\hat{P}}(x,y)$. Consequently, \eqref{eq:RCU_with_learning} becomes
\begingroup
\allowdisplaybreaks
\begin{align}\label{eq:RCU_with_learning2}
    \epsilon^*\leq& \min_{n_0\in\mathbb{Z}}\;  \Bigg\{\mathbb{P}\left[\sum_{i = 1}^{n_0}i_{\hat{P}}\left(X_i,Y_i\right) + \log U \leq \log{(L(\lceil2^{\frac{nR}{L}}\rceil^L-1))}  \right]\nonumber\\
    &+\sqrt{1-\exp\left(-\frac{n_0\left((|\mathcal{X}||\mathcal{Y}|-1)\ln (m+1) - \ln\delta\right)}{m}\right)}\Bigg\}\nonumber\\
    \leq& \mathbb{P}\left[\sum_{i = 1}^{n}i_{\hat{P}}\left(X_i,Y_i\right) + \log U \leq \log(2^{nR}-1)\right]\nonumber\\
    &+\sqrt{1-\exp\left(-\frac{n\left((|\mathcal{X}||\mathcal{Y}|-1)\ln (m+1) - \ln\delta\right)}{m}\right)}
\end{align}
\endgroup
The second inequality follows from setting $n_0 = n$. The RHS of \eqref{eq:RCU_with_learning2} should be less than $\epsilon$. Suppose we bound the first and second terms of the RHS of \eqref{eq:RCU_with_learning2} by $\epsilon-\frac{1}{\sqrt{n}}$ and $\frac{1}{\sqrt{n}}$, respectively. The bound for the second term tells us that the blocklength $n$ should satisfy
\begin{align*}
    \sqrt{1-\exp\left(-\frac{n\left((|\mathcal{X}||\mathcal{Y}|-1)\ln (m+1) - \ln\delta\right)}{m}\right)}\leq \frac{1}{\sqrt{n}},
\end{align*}
which, after some algebraic manipulation, gives us
\begin{align*}
n\leq&\frac{m\cdot \ln\left(\frac{n}{n-1}\right)}{(|\mathcal{X}||\mathcal{Y}|-1)\ln (m+1) - \ln\delta}.
\end{align*}
Since $\frac{1}{n}\leq \ln\left(\frac{n}{n-1}\right)$ for all $n > 1$, the above inequality will hold if $n$ satisfies condition \eqref{eq:blocklength_condition}.


We now apply Theorem \ref{theorem:berry_esseen} to bound the first term of \eqref{eq:RCU_with_learning2}. Let $T = \mathbb{E}\left[|i_{\hat{P}}(X_1,Y_1) - C_{\hat{P}}|^3\right]$, $\tilde{c} = \mathbb{E}[\log_2 U]$, $\tilde{v} = \mathrm{Var}[\log_2 U]$, and $\tilde{t} = \mathbb{E}\left[|\log_2 U - \tilde{c}|^3\right]$. Then, we get
\begingroup
\allowdisplaybreaks
\begin{align}\label{eq:apply_BE}
    &\mathbb{P}\left[\sum_{i = 1}^{n}i_{\hat{P}}\left(X_i,Y_i\right) + \log U \leq \log_2{(2^{nR}-1)} \right]\nonumber \\
    &\qquad\leq \mathcal{Q}\left(\frac{C_{\hat{P}} - R + \frac{\tilde{c}}{n}}{\sqrt{\frac{V_{\hat{P}}^{\epsilon}+\frac{\tilde{v}}{n}}{n}}}\right) + \frac{B(n)}{\sqrt{n}}\;\leq \epsilon-\frac{1}{\sqrt{n}}.
\end{align}
\endgroup
where $B(n) = 6\left(T+\frac{\tilde{t}}{n}\right)/\left(V_{\hat{P}}+\frac{\tilde{v}}{n}\right)^{\frac{3}{2}}$. Suppose $n$ is sufficiently large such that $\epsilon - \frac{ 1+B(n)}{\sqrt{n}} \in (0,1)$. After some algebraic manipulation, we get
\begin{align}
    R =& C_{\hat{P}} + \frac{\tilde{c}}{n} - \sqrt{\frac{V_{\hat{P}}^{\epsilon} + \frac{\tilde{v}}{n}}{n}}\mathcal{Q}^{-1}\left(\epsilon - \frac{1+B(n)}{\sqrt{n}}\right)\nonumber\\
    =& C_{\hat{P}} - \sqrt{\frac{V_{\hat{P}}^{\epsilon}}{n}}\mathcal{Q}^{-1}\left(\epsilon\right) + \mathcal{O}\left(\frac{1}{n}\right).
\end{align}
The second line follows from the Taylor expansion of $Q^{-1}(\cdot)$. All the terms of order $\frac{1}{n}$ are gathered in $\mathcal{O}\left(\frac{1}{n}\right)$.

For the converse part, we note that
\begin{align*}
    R \leq -\frac{1}{n}\log \left\{\textcolor{black}{\inf_{P_{X^n}}}\beta_{\max\{0,1-\epsilon-\kappa\}}\left(\hat{P}_{Y|X}^{n}\textcolor{black}{\hat{P}_{X^n}},\textcolor{black}{\hat{P}_{Y^n}\hat{P}_{X^n}}\right)\right\}
\end{align*}
by \eqref{eq:metaconverse_learning}. Here, we let \textcolor{black}{$Q_{Y^n} = \hat{P}_{Y^n} = \int \hat{P}_{Y|X}^n d\hat{P}_{X^n}$.} By \eqref{eq:blocklength_condition} and the definition of $\kappa$, we have $\kappa \leq \frac{1}{\sqrt{n}}$. Suppose $n$ is sufficiently large such that $\epsilon+\frac{B(n)+2}{\sqrt{n}}\in(0,1)$. By applying \cite[Lemma 58, Equation 344]{Polyanskiy:2010}, we get
\begingroup
\allowdisplaybreaks
\begin{align}
    R \leq& C_{\hat{P}} + \sqrt{\frac{V_{\hat{P}}^{\epsilon}}{n}}\mathcal{Q}^{-1}\left(1-\epsilon-\kappa - \frac{B(n)+1}{\sqrt{n}}\right) + \frac{\log n}{2n}\nonumber\\
    \leq& C_{\hat{P}} + \sqrt{\frac{V_{\hat{P}}^{\epsilon}}{n}}\mathcal{Q}^{-1}\left(1-\epsilon-\frac{B(n)+2}{\sqrt{n}}\right) + \frac{\log n}{2n}\nonumber\\
    =& C_{\hat{P}} + \sqrt{\frac{V_{\hat{P}}^{\epsilon}}{n}}\mathcal{Q}^{-1}\left(1-\epsilon\right) + \mathcal{O}\left(\frac{\log n}{n}\right)\nonumber\\
    =& C_{\hat{P}} -\sqrt{\frac{V_{\hat{P}}^{\epsilon}}{n}}\mathcal{Q}^{-1}\left(\epsilon\right) + \mathcal{O}\left(\frac{\log n}{n}\right)
\end{align}
\endgroup
The second line follows from the fact that $\mathcal{Q}^{-1}(x)$ is a decreasing function of $x$. Hence, an upper bound is obtained by replacing $\kappa$ with $\frac{1}{\sqrt{n}}$. The third line follows from the Taylor expansion of $Q^{-1}(\cdot)$. All the terms of order $\frac{\log n}{n}$ and $\frac{1}{n}$ are gathered in $\mathcal{O}\left(\frac{\log n}{n}\right)$. The fourth line follows from $\mathcal{Q}^{-1}(1-x) = -\mathcal{Q}^{-1}(x)$. The proof is completed by combining the achievability and converse results. We note that $\mathcal{O}\left(\frac{1}{n}\right)$ is also $\mathcal{O}\left(\frac{\log n}{n}\right)$.

%% file: appendix_D.tex
The proof closely follows that of Theorem \ref{theorem:normal_approx_achievability}. We modify \eqref{eq:RCU_with_learning2} to fit in our setup:
\begin{align*}
\epsilon^*\leq& \mathbb{P}\left[\sum_{i = 1}^{n_0}i_{\hat{P}}\left(X_i,Y_i\right) + \log U \leq \log{(L(\lceil2^{\frac{nR}{L}}\rceil^L-1))} \right]\nonumber\\
    &+\sqrt{1-\exp\left(-\frac{n_0\left((|\mathcal{X}||\mathcal{Y}|-1)\ln (m+1) - \ln\delta\right)}{m}\right)}\nonumber\\
\end{align*}
Suppose we bound the first and second terms of the RHS by $\epsilon-\frac{1}{\sqrt{n_0}}$ and $\frac{1}{\sqrt{n_0}}$, respectively. Applying similar analysis as in Appendix \ref{proof:normal_approx_achievability}, we arrive at the conclusion that $n_0$ should satisfy condition \eqref{eq:blocklength_condition2}.

Applying the Berry-Esseen CLT, we get
\begin{align*}
     &\mathbb{P}\left[\sum_{i = 1}^{n_0}i_{\hat{P}}\left(X_i,Y_i\right) + \log U \leq \log{L} + nR + \frac{L}{2^{\frac{nR}{L}}} \right]\nonumber \\
    &\leq \mathcal{Q}\left(\frac{\frac{n_0C_{\hat{P}}}{n} - R - \frac{\log L}{n} - \frac{L\cdot 2^{-\frac{nR}{L}}}{n}+ \frac{\tilde{c}}{n}}{\sqrt{\frac{n_0V_{\hat{P}}^{(1-\alpha)\epsilon}+\tilde{v}}{n^2}}}\right) + \frac{B(n_0)}{\sqrt{n_0}}\\
    &\leq \mathcal{Q}\left(\frac{\frac{n_0C_{\hat{P}}}{n} - R - \frac{\log L}{n} - \frac{L}{n}+ \frac{\tilde{c}}{n}}{\sqrt{\frac{n_0V_{\hat{P}}^{(1-\alpha)\epsilon}+\tilde{v}}{n^2}}}\right) + \frac{B(n_0)}{\sqrt{n_0}}\\
    &\leq \epsilon - \frac{1}{\sqrt{n_0}},
    \end{align*}
where the other variables are already defined in Appendix \ref{proof:normal_approx_achievability}. The second inequality is obtained by noting that Q function is monotonic decreasing. Hence, an upper bound is obtained using $\frac{L\cdot 2^{-\frac{nR}{L}}}{n}\leq \frac{L}{n}$. Suppose $n_0$ is sufficiently large such that $\epsilon - \frac{B(n_0)+1}{\sqrt{n_0}}$ is between 0 to 1. After performing some algebraic manipulation similar to Appendix \ref{proof:normal_approx_achievability}, we get 
\begin{align*}
    R_{\mathrm{achievable}}^* =& \frac{n_0 C_{\hat{P}}}{n} - \sqrt{\frac{n_0V_{\hat{P}}^{\epsilon}}{n^2}}\mathcal{Q}^{-1}\left(\epsilon\right)\\
    &\quad-\frac{L + \log L}{n} + \mathcal{O}\left(\frac{1}{n}\right)\\
    =& \frac{n_0 C_{\hat{P}}}{n} - \sqrt{\frac{n_0V_{\hat{P}}^{\epsilon}}{n^2}}\mathcal{Q}^{-1}\left(\epsilon\right) + \mathcal{O}\left(\frac{1}{n}\right),
\end{align*}
where the $-\frac{L+\log L}{n}$ term is absorbed by $\mathcal{O}(\frac{1}{n})$ in the second line.
